\documentclass[a4paper,runningheads]{llncs}

\usepackage{amssymb}
\usepackage{amsmath}
\usepackage{amsfonts}
\usepackage{graphicx}
\RequirePackage{fancyhdr}
\usepackage{xcolor}
\usepackage{boxedminipage}
\usepackage{tikz}
\usetikzlibrary{positioning,arrows,automata,shapes}
\usepackage{url}
\usepackage{subfigure}
\usepackage[linesnumbered, boxed, algosection]{algorithm2e}

\newcommand{\name}[1]{\textsc{#1}}

\newcommand{\prepc}{{\sf PrepC}}
\newcommand{\prepp}{{\sf PrepP}}
\newcommand{\appc}{{\sf AppC}}
\newcommand{\appp}{{\sf AppP}}

\newcommand{\parnamedefn}[4]{
  \begin{tabbing}
    \name{#1} \\

    \emph{Input:} \hspace{0.8cm} \= \parbox[t]{\textwidth}{#2} \\
    \emph{Parameter:}            \> \parbox[t]{\textwidth}{#3} \\
    \emph{Question:}             \> \parbox[t]{\textwidth}{#4} \\
    \end{tabbing}
  \vspace{-0.5cm}
}

\newcommand{\decnamedefn}[3]{
  \begin{tabbing} #1\\

    \emph{Input:} \hspace{0.8cm} \= \parbox[t]{\textwidth}{#2} \\
    \emph{Question:}             \> \parbox[t]{\textwidth}{#3} \\
  \end{tabbing}
}

\newcommand{\ETH}{\textsf{ETH}}
\newcommand{\Yes}{{\sc Yes}}

\newcommand{\subdagiso}{{\sc SubTDAG isomorphism}}

\newcommand{\wsp}{{\sc WSP}$(=,\neq,<)$}

\newtheorem{assumption}{Assumption}
\newtheorem{observation}{Observation}

\title{Fixed-Parameter Tractability of\\ Workflow Satisfiability in the\\ Presence of Seniority Constraints}
\author{
J. Crampton\inst{1} \and R. Crowston\inst{1} \and G. Gutin\inst{1} \and M. Jones\inst{1} \and M.S. Ramanujan\inst{2}
}
\institute{Royal Holloway University of London, United Kingdom\\
           \and
           The Institute of Mathematical Sciences, Chennai, India}

\begin{document}

\maketitle

\begin{abstract}
The workflow satisfiability problem is concerned with determining whether it is possible to find an allocation of authorized users to the steps in a workflow in such a way that all constraints are satisfied.  The problem is NP-hard in general, but is known to be fixed-parameter tractable for certain classes of constraints.  The known results on fixed-parameter tractability rely on the symmetry (in some sense) of the constraints.  In this paper, we provide the first results that establish fixed-parameter tractability of the satisfiability problem when the constraints are asymmetric.  In particular, we introduce the notion of seniority constraints, in which the execution of steps is determined, in part, by the relative seniority of the users that perform them.  Our results require new techniques, which make use of tree decompositions of the graph of the binary relation defining the constraint.  Finally, we establish a lower bound for the hardness of the workflow satisfiability problem.
\end{abstract}

\setcounter{footnote}{0}

\section{Introduction}\label{sec:intro}

A business process is a collection of interrelated steps that are performed in some predetermined sequence in order to achieve some objective.
It is increasingly common to automate business process and for business process management systems or workflow management systems to control the execution of the steps comprising the business process.
A workflow specification is an abstract representation of a collection of business steps, together with dependencies on the order in which steps should be performed.
A workflow specification may be instantiated and its execution controlled by a workflow management system.

In many situations, we wish to restrict the users that can perform certain steps.
On the one hand, we may wish to specify which users are authorized to perform particular steps.
The workflow management system will prevent a user from performing any step for which that user is not authorized.
In addition, we may wish, either because of the particular requirements of the business application or because of statutory requirements, to prevent certain combinations of users from performing particular combinations of steps.
In particular, there may be pairs of steps that must be executed in any given instance of the workflow by different users, the so-called ``two-man rule'' (or ``four-eyes rule'').
Similarly, we may require that two or more steps in any given instance are performed by the same user.
These constraints are sometimes known in the literature as separation-of-duty and binding-of-duty constraints, respectively.

The existence of constraints on the execution of a workflow raises the question of whether a workflow specification can be realized in practice.
As a trivial example, a workflow with two steps and the requirement that a different user performs each of the two steps cannot be realized by a user population with a single user.
Therefore, it is important to be able to determine whether a workflow is satisfiable: Does there exist an allocation of authorized users to workflow steps such that every step is performed by an authorized user and are all constraints on the execution of steps satisfied?

A brute-force approach to answering the question gives rise to an algorithm that has running time $O(cn^k)$, where $c$ is the number of constraints\footnote{Here and in the rest of the paper, all constraints are binary and a constraint can be checked in constant time.}, $n$ is the number of users and $k$ is the number of steps.
Moreover, it is known that determining the satisfiability of a workflow specification is NP-hard in general~\cite{WangLi10}.
However, it has also been shown that some interesting special cases of the problem are fixed-parameter tractable, meaning that there exists an algorithm to solve them with running time $O(f(k)n^d)$, where $d$ is some constant (independent of $k$ and $n$).
The existence of such an algorithm suggests that relatively efficient methods can be developed to solve interesting cases of the workflow satisfiability problem.

Wang and Li established that satisfiability is fixed-parameter tractable when we restrict attention to separation- and binding-of-duty constraints~\cite{WangLi10}.
Crampton {\em et al.} developed a novel analysis of the problem, which reduced the complexity considerably, but retained the focus on separation- and binding-of-duty constraints~\cite{CrGuYe12}.
In this paper, we consider a new class of constraints, in which the users that perform two steps are different and one is senior to the other.
Seniority constraints are asymmetric, in contrast to separation- and binding-of-duty constraints, and this means that existing techniques for determining workflow satisfiability cannot be applied to workflow specifications that contain such constraints.

In this paper, we introduce novel techniques for determining workflow satisfiability when the specification includes seniority constraints.
These techniques are based on the tree decomposition of the graph of the seniority relation and the application of dynamic programming to a particular form of tree decomposition.
This enables us to establish that the workflow satisfiability problem is fixed-parameter tractable when the  partial order defined over the set of users has Hasse diagram (viewed as an acyclic digraph) of bounded treewidth\footnote{We define treewidth of a graph in Sec.~\ref{sec:tree-decomposition}.}. As we will see, many user hierarchies that arise in practice have bounded treewidth.  However, our result is highly unlikely to hold for an arbitrary partial order defined over the set of users.
Moreover, we show that it is impossible to obtain an algorithm for the general case of WSP with running time significantly better than $O(cn^k)$, assuming the Exponential Time Hypothesis (ETH)~\cite{ImPaZa01} holds.

We conclude this section by providing some terminology and notation on directed and undirected graphs.
In the next section, we introduce the workflow satisfiability problem and further justify the relevance of seniority constraints.
In Sec.~\ref{sec:tree-decomposition}, we describe tree decompositions, define treewidth and show its relevance to practical  seniority constraints, and establish some elementary, preparatory results.
Section~\ref{sec:bounded-treewidth} establishes fixed-parameter tractability of the above-mentioned ``treewidth bounded" case of the problem and the following section establishes a lower bound for the complexity of the general problem (assuming ETH holds).
We conclude the paper with a summary of our contributions, a discussion of the significance of our results, and some suggestions for future work.

\paragraph{Terminology and Notation for Graphs} Let $G$ be a directed or undirected graph and let $X$ be a set of vertices of $G$. The subgraph $G[X]$ of $G$ {\em induced by} $X$ is obtained from $G$ by deleting all vertices not in $X$. Let $D$ be a directed graph. The {\em underlying graph} $U(D)$ is the undirected graph obtained from $D$ by removing orientations from all arcs of $D$. We say that $D$ is {\em connected} if $U(D)$ is connected.
We say that $D$ is {\em transitive} if for every pair $x,y$ of distinct vertices, if there is a directed path from $x$ to $y$ then $D$ contains an arc from $x$ to $y$. We say that a directed graph $H$ is the {\em transitive closure} of $D$
if there is an arc from $x$ to $y$ in $H$ whenever there is a directed path from $x$ to $y$ in $D$.
The {\em degree} of a vertex $x$ of $D$ is its degree in $U(D)$. Let $H$ be a directed or undirected graph. For a natural number $\ell$, we say that $H$ is $\ell$-degenerate if $H[X]$ has a vertex of degree at most $\ell$ for each set of vertices $X$ of $H$. As an example, consider a forest. Note that it is 1-degenerate.
Let $D$ be a digraph, $Y$ a set of vertices of $D$, and $y,z$ vertices in $D - Y$. We say that $Y$ {\em separates} $y$ {\em from} $z$ if $D - Y$ has no directed path from $y$ to $z$.

\section{Workflow Satisfiability}\label{sec:satisfiability}

Suppose we are given a workflow specification comprising a set $S$ of $k$ steps.
A workflow constraint has the form $(\rho,s',s'')$, where $s',s'' \in S$ and $\rho$ is a binary relation defined over a set $U$ of $n$ users.
For each step $s\in S$, there is a list $L(s)$ of users authorized to perform $s$.
A function $\pi$ from $S$ to $U$ is called a {\em plan}. We say that a plan $\pi$  {\em satisfies} constraint $(\rho,s',s'')$ if $(\pi(s'),\pi(s''))\in \rho$.

For a set, $\{\rho_1,\ldots ,\rho_t\}$, of binary relations on $U$, an instance $\cal I$ of the workflow satisfiability problem {\sc WSP}($\rho_1,\ldots ,\rho_t$) is given by a list $L(s)$ for each $s\in S$ and a set $C$ of constraints of the form $(\rho,s',s'')$, where $s',s''\in S$ and $\rho \in \{\rho_1,\dots,\rho_t\}$; we are to decide whether there is a {\em valid} plan, i.e., a plan $\pi$ such that the following hold:
\begin{itemize}
\item for each $s\in S$, $\pi(s)\in L(s)$;
\item $\pi$ satisfies each constraint $(\rho_i,s',s'')\in C$.
\end{itemize}
If $\cal I$ has a valid plan, it is called a {\sc Yes}-instance. Otherwise, it is a {\sc No}-instance.

Let $<$ be a partial order on $U$.
We will consider constraints of the form $(\rho,s',s'')$, where $\rho$ is one of $=$, $\neq$, $<$, and $s',s'' \in S$.
A plan $\pi$ satisfies:
\begin{itemize}
  \item $(=,s',s'')$ if $\pi(s')=\pi(s'')$;
  \item $(\neq,s',s'')$ if $\pi(s')\neq \pi(s'')$;
  \item $(<,s',s'')$ if $\pi(s')< \pi(s'')$.
\end{itemize}

Consider a business process for handling expenses claims, which
is illustrated in Fig.~\ref{fig:example-workflow}.
Such a workflow might include four steps: the preparation of an expenses claim (\prepc), the approval of the claim (\appc), the preparation of the payment (\prepp), and the approval of the payment (\appp).
We might assume that most, if not all, users in an organization are authorized to prepare an expenses claim.
We require that the user who approves a claim is senior to the user who prepares a claim.
Note that it would be either difficult or impractical to enforce this rule simply by restricting the users who are authorized to approve claims.
(We could authorize only the most senior user to approve expenses claims, but this is unnecessarily limiting and places an onerous burden on a single individual.)
Similarly, we require that the user who approves a payment be senior to the user who prepares the payment.
In addition, we require that the user who prepares the expenses claim is different from the one who prepares the payment, and the user that approves the claim is different from the user who prepares the payment and from the one who approves the payment.

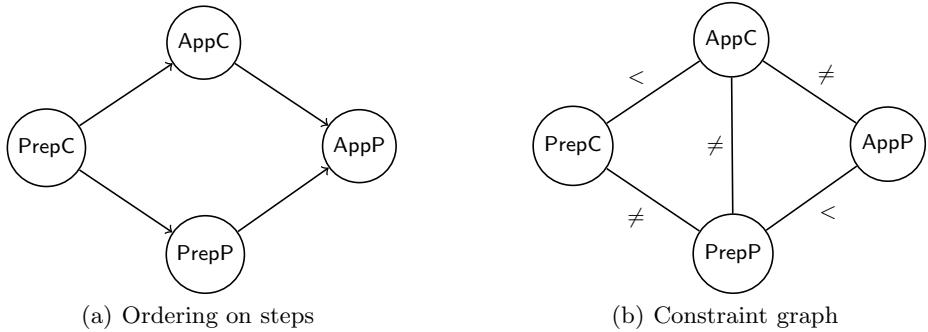
\begin{figure}[ht]\centering
\subfigure[Ordering on steps]{
\begin{tikzpicture}[->,.=stealth',node distance=0.75cm and 1.5cm,semithick,auto,inner sep=1mm,scale=.9,transform shape]
  \node[state]  (t1)                      {\prepc};
  \node[state]  (t2) [above right=of t1]        {$\appc$};
  \node[state]  (t3) [below right=of t1]  {\prepp};
  \node[state]  (t4) [below right=of t2]  {\appp};
  \path (t1) edge (t2)
   (t1) edge (t3)
   (t2) edge (t4)
   (t3) edge (t4);
\end{tikzpicture}}
\hfill
\subfigure[Constraint graph]{
\begin{tikzpicture}[-,.=stealth',node distance=0.75cm and 1.5cm,semithick,auto,scale=.9,transform shape]
  \node[state]  (t1)                      {\prepc};
  \node[state]  (t2) [above right=of t1]        {$\appc$};
  \node[state]  (t3) [below right=of t1]  {\prepp};
  \node[state]  (t4) [below right=of t2]  {\appp};
  \draw[-] (t1) to node {$<$} (t2);
  \draw[-] (t1) to node[swap] {$\ne$} (t3);
  \draw[-] (t2) to node[swap] {$\ne$} (t3);
  \draw[-] (t2) to node {$\ne$} (t4);
  \draw[-] (t3) to node[swap] {$<$} (t4);
\end{tikzpicture}}
\caption{A simple constrained workflow for purchase order processing}\label{fig:example-workflow}
\end{figure}

It is perhaps worth noting at this stage that the use of an access control model that incorporates some notion of seniority (role-based access control and information flow models being obvious candidates) does not necessarily enforce the desired constraints.
We might assign the \prepc\ and \appc\ steps to two different roles $r$ and $r'$, say, with $r < r'$.
However, this does not enforce the desired constraint: a user assigned to $r'$ is indirectly assigned to $r$ and is, therefore, authorized to perform both steps.

It is worth noting, however, that access control models do define (albeit indirectly) an ordering on the set of users.
In particular, we may define $u < u'$ if the set of steps for which $u$ is authorized is a strict subset of the set of steps for which $u'$ is authorized.
The relation $<$ is transitive.
The relation $\leq$, where $u \leq u'$ if and only if $u < u'$ or $u = u'$ is transitive, reflexive and anti-symmetric; that is, $\leq$ defines a partial order on $U$.

We also note that many organizations have user hierarchies that define the reporting and management lines within those organizations.
If such a hierarchy exists, we may evaluate our seniority constraints with respect to such a hierarchy (rather than an ordering defined by the authorization policy).
In many cases, such a user hierarchy will be a rooted tree, although our results do not require this and more complex hierarchies do arise in practice.
At Royal Holloway, University of London, for example, each of the three faculty Deans reports to and is managed by each of the three Vice Principals, as shown in Fig.~\ref{fig:example-rhul}.
The complete bipartite subgraph within a user hierarchy that is a feature of this hierarchy also arises in the (graphs of the) relations of the preorders that are obtained from an authorization policy: each user in the set of users authorized for $S' \subseteq S$ is senior to each user in the set of users authorized for $S'' \subset S'$.

\begin{figure}\centering
\begin{tikzpicture}
\node (P) at (2,-1) {Principal};
\node (VP1) at (0,-2) {VP$_1$} edge [->] (P);
\node (VP2) at (2,-2) {VP$_2$} edge [->] (P);
\node (VP3) at (4,-2) {VP$_3$} edge [->] (P);
\node (D1) at (0,-4) {Dean$_1$} edge [->] (VP1) edge [->] (VP2) edge [->] (VP3);
\node (D2) at (2,-4) {Dean$_2$} edge [->] (VP1) edge [->] (VP2) edge [->] (VP3);
\node (D3) at (4,-4) {Dean$_3$} edge [->] (VP1) edge [->] (VP2) edge [->] (VP3);
\end{tikzpicture}
\caption{Part of the user hierarchy at Royal Holloway}\label{fig:example-rhul}
\end{figure}
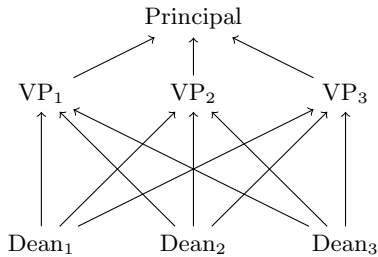

\subsection{Constraint Graphs}

Given a partial (irreflexive) order  $<$ on  $U$, let $H$ be the transitive acyclic graph with vertex set $U$ such that $u < v$ if and only if there is an arc from $u$ to $v$ in $H$. We say $H$ is the \emph{full graph} of $(U, <)$.
Let $D$ be an directed acyclic graph such that $H$ is the transitive closure of $D$ and the transitive closure of every subgraph $D-a$, where $a$ is an arc of $D$, is not equal to $H$. Note that since $H$ is acyclic, $D$ is unique \cite{AhoGarUll72} (see also Sec. 2.3 of \cite{BJBGG}).
We say that $D$ is the \emph{reduced graph} (or \emph{Hasse diagram}) of $(U,<)$.

A \emph{mixed graph} consists of a set of vertices together with a set of undirected edges and a set of directed arcs.
We may represent the set of constraints with a mixed graph as follows.

First, we eliminate constraints of the form $(=,s',s'')$.
Specifically, we construct a graph $P$ with vertices $S$ in which $s',s''\in S$ are adjacent
if $\cal I$ has a constraint $(=,s',s'')$. Observe that the same user must necessarily be assigned to all steps
in a connected component $Q$ of $P$. Thus,
if there is a pair $s',s''\in Q$ such that $\cal I$ has a constraint
$(\neq ,s',s'')$ or $(<,s',s'')$, then clearly $\cal I$ is a {\sc
No}-instance; thus we may assume
that there is no such pair for any connected component of $H$.
For each connected component $Q$ of $P$, replace all steps of $Q$ in $S$
by a ``superstep'' $q$. A user $u$ is authorized to perform $q$ if $u$ is
authorized to perform all steps of $Q$. That is, $L(q)=\bigcap_{s\in Q}
L(s)$.

The above procedure eliminates all constraints of the type $(=,s',s'')$ for
the reduced set $S$ of steps. All constraints of the types $(\neq ,s',s'')$ and $(<,s',s'')$
remain, but steps $s'$ and $s''$ are replaced by the corresponding
supersteps. For simplicity of notation, we will denote the new instance of
the problem also by $\cal I$.

Now we construct a mixed graph with vertex set $S$. For each constraint of the type $(\neq ,s',s'')$, add an \emph{edge} between $s'$ and $s''$. For each constraint of the type $(< ,s',s'')$, add an \emph{arc} from $s'$ and $s''$.
We will refer to the resulting graph as the \emph{constraint graph} (of $\mathcal{I}$).
We will say an edge or arc in a constraint graph is \emph{satisfied} by a
plan $\pi$ if $\pi$ satisfies the corresponding constraint.

It is worth noting that WSP$(\neq)$ is rather closely related to graph colorability, where the assignment of users to tasks in such a way that separation-of-duty constraints are satisfied provides a coloring of the constraint graph and vice versa\footnote{In fact, WSP($\neq$) is equivalent to the more general problem List Coloring, as the list $L(s)$ imposes restrictions on the ``colors'' (users) that can be assigned to step $s$.}.
Note that the selection of a color for step $s$ in the constraint graph prevents the use of only one color for steps connected by an edge to $s$.  WSP$(<,\ne)$ is an even more complex problem because it imposes a structure on the set of colors that are available, meaning that the selection of a color for $s$ may preclude the use of many other colors for steps connected to $s$ by an arc.

Consider, for example, an organization with three users -- Alice, Bob and Carol, where Alice is senior to Bob and Carol and all three users are authorized for all tasks.
Then, our expenses claim workflow is not satisfiable.
However, the workflow specification is satisfiable if we replace the seniority constraints with separation-of-duty constraints.%

\subsection{Related Work}

Suppose we have an algorithm that solves an NP-hard problem in time $O(f(k)n^d)$, where $n$ denotes the size of the input to the problem, $k$ is some (small) parameter of the problem, $f$ is some function in $k$ only, and $d$ is some constant (independent of $k$ and $n$).
Then we say the algorithm is a \emph{fixed-parameter tractable} (FPT) algorithm.
If a problem can be solved using an FPT algorithm then we say that it is an \emph{FPT problem} and that it belongs to the class FPT\footnote{For more information on parameterized algorithms and complexity, see monographs \cite{DowFel99,FluGro06,Nie06}.}.

Wang and Li initiated the study of the fixed-parameter tractability of  workflow satisfiability~\cite{WangLi10}.
They showed that the problem is W[1]-hard, in general, which implies that it is not FPT (unless the parameterized complexity hypothesis $\text{FPT}\neq\text{W[1]}$ fails, which is believed to be highly unlikely).
However, they were able to show that WSP$(=,\ne)$ is FPT.

Crampton {\em et al.} introduced new techniques for analyzing WSP$(=,\ne)$ and significantly improved the complexity of FPT algorithms to solve the problem~\cite{CrGuYe12}.
The approach of Crampton {\em et al.} is based on partitioning the set of steps and, for each block of steps in the partition, assigning a user to that block, where the user was authorized for each step in the block.
The existence of such a partition and allocation of users to blocks demonstrates that a workflow specification is satisfiable.
This method assumes that the allocation of a user to one particular block is independent of the allocation of users to other blocks: this assumption holds for separation- and binding-of-duty constraints; however, it does not hold for seniority constraints because the choice of a senior user for one block may limit the choices of user available for other blocks.

Constraints of the form $(\rho,s',s'')$ have been called Type 1 constraints~\cite{CrGuYe12}, and were formally introduced by Crampton~\cite{Cr05}.
Wang and Li introduced Type 2 constraints~\cite{WangLi10}, which have the form $(\rho,s',S')$, where $S' \subseteq S$ and the constraint is satisfied by plan $\pi$ if there exists $s'' \in S'$ such that $(\pi(s'),\pi(s'')) \in \rho$.
Finally, Crampton {\em et al.} defined Type 3 constraints~\cite{CrGuYe12}, which have the form $(\rho,S',S'')$, where $S',S'' \subseteq S$ and the constraint is satisfied if there exist $s' \in S'$ and $s'' \in S''$ such that $(\pi(s'),\pi(s'')) \in \rho$.

Crampton {\em et al.} \cite{CrGuYeJournal} showed that it is possible to rewrite a workflow specification containing Type 2 or Type 3 constraints as a collection of workflow specifications, each containing Type 1 constraints only.
Moreover, the number of workflow specifications is determined by $k$ (the number of steps) only, which means that the existence of an FPT algorithm for Type 1 constraints can be used to establish the existence of an FPT algorithm for specifications containing any combination of Type 1, 2 or 3 constraints.
In this paper, we demonstrate the existence of an FPT algorithm for Type 1 constraints containing the $<$ relation provided the reduced graph $D$ is of bounded treewidth.
The prior work of Crampton {\em et al.} \cite{CrGuYeJournal}  enables us to construct an FPT algorithm for Type 2 and 3 constraints.

\section{Tree Decompositions and Treewidth}\label{sec:tree-decomposition}

Tree decompositions provide a means of representing a (directed) graph using a tree.
Subsets of the graph's vertices form the nodes of the tree, in such a way that a subtree containing a particular vertex is connected and the subtrees associated with the end-points of an edge in the graph have nonempty intersection.
The treewidth of a graph $G$ is a measure of the minimum number of vertices that are required in each node of a tree in order to construct a tree decomposition of $G$.
Treewidth is known to be an important parameter when considering the complexity of graph-related problems that are NP-hard in general.
As we will see, treewidth plays an important role in the complexity of the workflow satisfiability problem when we define a transitive relation $<$ on $U$ and define workflow constraints in terms of $<$.

\begin{definition}\label{TDdef}
 A \emph{tree decomposition} of a (directed) graph $G = (V,E)$ is a pair $(\mathcal{T}, \mathcal{X})$, where $\mathcal{T} = (V_\mathcal{T}, E_\mathcal{T})$ is a tree and $\mathcal{X} = \{ \mathcal{B}_i: i \in V_\mathcal{T}\}$ is a collection of subsets of $V$ called \emph{bags}, such that
\begin{enumerate}
 \item $\bigcup_{i \in V_\mathcal{T}} \mathcal{B}_i  = V$.
 \item For every edge (arc) $xy \in E$, there exists $i \in V_\mathcal{T}$ such that $\{x,y\} \subseteq \mathcal{B}_i$.
 \item For every $x \in V$, the set $\{i: x \in \mathcal{B}_i\}$ induces a connected subtree of $\mathcal{T}$.
\end{enumerate}
The \emph{width} of $(\mathcal{T}, \mathcal{X})$ is $\max_{i \in V_\mathcal{T}}|\mathcal{B}_i|-1$. The \emph{treewidth} of a graph $G$ is the minimum width of all tree decompositions of $G$.
\end{definition}

To distinguish between vertices of $G$ and $\cal T$, we call vertices of $\cal T$ {\em nodes}.
We will often speak of a bag $\mathcal{B}$ interchangeably with the node it corresponds to in $\mathcal{T}$. Thus, for example, we might say two bags $\mathcal{B}, \mathcal{B}'$ are neighbors if they correspond to nodes in $\mathcal{T}$ which are neighbors.
We define the \emph{descendants} of a bag $\mathcal{B}$ as follows: every child of $\mathcal{B}$ is a descendant of $\mathcal{B}$, and every child of a descendant of $\mathcal{B}$ is a descendant of $\mathcal{B}$.
At the same time, we will say $\mathcal{B} = \mathcal{B'}$ if $\mathcal{B}, \mathcal{B}'$ contain the same vertices, while still treating them as different bags.

\vspace{2 mm}

It is well-known that a connected graph is of treewidth 1 if and only if it is a tree with at least one edge \cite{Klo94}. Every tree $T$ with at least one edge has the following tree decomposition $\cal T$ of width 1:
for every vertex $x$ of $T$ let $\{x\}$ be a bag of $\cal T$ and for every edge $xy$ of $T$ let $\{x,y\}$ be a bag
of $\cal T$. Two bags are adjacent in $\cal T$ if one of them is a proper subset of the other.
For the graph depicted in Fig. \ref{fig:example-workflow} (b) there is a tree decomposition of width 2: it has two bags $\{\appc,\prepc,\prepp\}$ and $\{\appc,\appp,\prepp\}$ connected by an edge.

The graph of Fig.~\ref{fig:example-rhul} has a tree decomposition of width 3, as shown in Fig.~\ref{fig:example-rhul-tree}.
The graph of Fig.~\ref{fig:example-rhul}  can be extended as follows to more fully reflect the Royal Holloway management hierarchy. Each faculty at Royal Holloway has several academic departments each led by Head of Department (HoD) and so we may add HoD's, each with an arc to the corresponding Dean, and non-HoD members of staff, each with an arc to the corresponding HoD. This extension of the graph of Fig.~\ref{fig:example-rhul} essentially adds just trees to the graph and it is not hard to check that the treewidth of the extended graph is still 3.

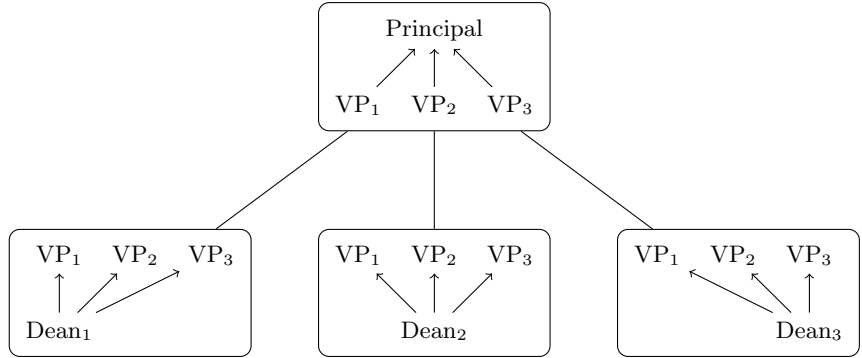
\begin{figure}\centering
\begin{tikzpicture}
\node[shape=rectangle,rounded corners,draw] (Root) at (0,0){
\begin{tikzpicture}
\node (P) at (1,0) {Principal};
\node (VP1) at (0,-1) {VP$_1$} edge [->] (P);
\node (VP2) at (1,-1) {VP$_2$} edge [->] (P);
\node (VP3) at (2,-1) {VP$_3$} edge [->] (P);
\end{tikzpicture}
};

\node[shape=rectangle,rounded corners,draw] (D1Bag) at (-4,-3){
\begin{tikzpicture}
\node (VP1) at (0,0) {VP$_1$};
\node (VP2) at (1,0) {VP$_2$};
\node (VP3) at (2,0) {VP$_3$};
\node (D1) at (0,-1) {Dean$_1$} edge [->] (VP1) edge [->] (VP2) edge [->] (VP3);
\end{tikzpicture}
};

\node[shape=rectangle,rounded corners,draw] (D2Bag) at (0,-3){
\begin{tikzpicture}
\node (VP1) at (0,0) {VP$_1$};
\node (VP2) at (1,0) {VP$_2$};
\node (VP3) at (2,0) {VP$_3$};
\node (D2) at (1,-1) {Dean$_2$} edge [->] (VP1) edge [->] (VP2) edge [->] (VP3);
\end{tikzpicture}
};

\node[shape=rectangle,rounded corners,draw] (D3Bag) at (4,-3){
\begin{tikzpicture}
\node (VP1) at (0,0) {VP$_1$};
\node (VP2) at (1,0) {VP$_2$};
\node (VP3) at (2,0) {VP$_3$};
\node (D3) at (2,-1) {Dean$_3$} edge [->] (VP1) edge [->] (VP2) edge [->] (VP3);
\end{tikzpicture}
};

\draw (Root) -- (D1Bag);
\draw (Root) -- (D2Bag);
\draw (Root) -- (D3Bag);
\end{tikzpicture}
\caption{Tree Decomposition of Royal Holloway management hierarchy}\label{fig:example-rhul-tree}
\end{figure}

The Royal Holloway management hierarchy is not exceptional in the following sense: it is unlikely that a member of staff will have many line managers (quite often there is only one line manager). Thus, it does not seem unreasonable to expect the reduced graph of the corresponding partial order to have bounded treewidth and for the treewidth to be rather small. Moreover, our Royal Holloway example indicates that construction of (near-)optimal tree decompositions for such hierarchies may be not hard.

\vspace{2mm}

It is NP-complete to decide whether the treewidth of a graph $G$ is at most $r$ (when $r$ is part of input) \cite{ArnCorPro87}.
Bodlaender \cite{Bod96} obtained an algorithm with running time $O(f(r)n)$ for deciding whether the treewidth of a graph $G$ is at most $r$, where $n$ is the number of vertices in $G$ and $f$ is a function depending only on $r$. This algorithm constructs the corresponding tree decomposition with $O(n)$ nodes, if the answer is {\sc Yes}. Unfortunately, $f$ grows too fast to be of practical interest. However, there are several polynomial-time approximation algorithms and heuristics for computing the treewidth of a graph  and the corresponding tree decomposition, see, e.g., \cite{BodKos10}.

%
%
%
%
%
%

We now describe a special type of tree decomposition
that is widely used to construct dynamic programming
algorithms for solving problems on graphs, called a {\em nice tree decomposition}.
In a nice tree decomposition, one node in $\cal T$ is considered
to be the root of $\cal T$, and each node $i \in V_\mathcal{T}$ is of one of the
following four types:

\begin{enumerate}
 \item a {\em join node} $\mathcal{B}$ has two children $\mathcal{B}'$ and $\mathcal{B}''$, with $\mathcal{B}=\mathcal{B}'=\mathcal{B}''$;
 \item a {\em forget node} $\mathcal{B}$ has one child $\mathcal{B}'$, and there exists $u \in \mathcal{B}'$ such that  $\mathcal{B} = \mathcal{B}' \setminus \{u\}$;
 \item an {\em introduce node} $\mathcal{B}$ has one child $\mathcal{B}'$, and there exists $u \not\in \mathcal{B}'$ such that $\mathcal{B}= \mathcal{B}' \cup \{u\}$;
 \item a {\em leaf node} $\mathcal{B}$ is a leaf of $\cal T$.
\end{enumerate}

The following useful lemma, concerning the construction of a nice tree decomposition from a given tree decomposition, was proved by Kloks \cite[Lemma 13.1.3]{Klo94}.

 \begin{lemma} \label{lem:niceform}
Given a tree decomposition with $O(n)$ nodes of a graph $G$ with $n$ vertices, we can construct, in time $O(n)$, a nice tree decomposition of $G$ of the same width and with at most $4n$ nodes.
\end{lemma}



\begin{lemma}\label{lem:connected}
Let $D$ be a (directed) graph, $(\mathcal{T}, \mathcal{X})$ a tree decomposition of $D$, and let $Y$ be a set of vertices in $D$ such that $D[Y]$ is connected. Then the set of bags containing vertices in $Y$ induces a connected subtree in $\mathcal{T}$.
\end{lemma}
\begin{proof}
The proof is by induction on $|Y|$. The base case, $|Y|=1$, follows from Definition \ref{TDdef}. Let $y\in Y$
such that $D[Y\setminus \{y\}]$ is connected and suppose that the set of bags containing vertices in $Y\setminus \{y\}$ induces a connected subtree $\mathcal{T}'$ of $\mathcal{T}$. Let $z\in Y$ such that $yz$ is an edge of $D$. By Definition \ref{TDdef}, $y$ and $z$ belong to the same bag $\cal B$ and observe that  $\cal B$ is in $\mathcal{T}'$. Thus, the subtree of $\cal T$ induced by the bags containing $y$ and $\mathcal{T}'$ intersect and so the set of bags containing vertices in $Y$ induces a connected subtree in $\mathcal{T}$.\qed
\end{proof}

\begin{lemma}\label{lem:separator}
 Let $D$ be the reduced graph for $(U,<)$. Let $u,v$ be users and $\mathcal{B}$ a set of users such that $u \neq v$ and $u,v \notin \mathcal{B}$, and $\mathcal{B}$ separates $u$ from $v$ in $D$.  Then $u < v$ if and only if there exists $w \in \mathcal{B}$ such that $u < w$ and $w < v$.
\end{lemma}
\begin{proof}
 By transitivity, if $u<w<v$ then $u<v$. For the other direction, suppose $u<v$. Then by the definition of $D$ there must exist a directed path from $u$ to $v$ in $D$. Since $\mathcal{B}$ separates $u$ and $v$, this path must contain a user $w$ in $\mathcal{B}$. Therefore $u < w$ and $w < v$.\qed
\end{proof}

%

\section{FPT Algorithm for Bounded Treewidth}\label{sec:bounded-treewidth}

In this section, we consider the special case of the problem when the reduced graph $D$ of $(U,<)$ is of bounded treewidth. In other words, in this section, we assume that the treewidth of $D$ is bounded by a constant $r$. Note that $D$ may have much smaller treewidth than the full graph $H$. For example, when $<$ is a linear order on $U$, then $H$ is a tournament with treewidth $|U|-1$, but $D$ is a directed path with treewidth $1$.

\begin{theorem}\label{thm:treewidth}
Let $\cal I$ be an instance of \wsp\ and let $D$ be the reduced graph of  $(U,<)$. Given a tree decomposition of $D$ of treewidth $r$ and with $O(n)$ nodes, we can solve $\cal I$ in time  $O(nk4^k(r+2+3^{r+1})^k)$, where $k$ is the number of steps and $n$ is the number of users.
\end{theorem}

By Lemma~\ref{lem:niceform}, assume we have a nice tree decomposition  $(\mathcal{T}, \mathcal{X})$ of $D$ of width $r$ and with at most $4n$ nodes.
Henceforth, we assume that we have constructed a nice tree decomposition for the instance $\cal I$.

Before proving the above result, we provide an informal insight into our approach.
Dynamic programming is a well known technique that is used to solve a problem by systematically solving subproblems, each of which may contribute to the solution of other (typically larger or more complex) subproblems. For example, one might solve all subproblems of size $i$, and use these to solve all subproblems of size $i+1$, or one might make use of structural graph properties, such as tree decompositions.

In the case of {\wsp} we use dynamic programming techniques to compute solutions to restricted instances of the original problem instance, and for each of these restricted instances, we construct possible intermediate solutions for each bag in the nice tree decomposition.
Working from the leaves of the decomposition back to the root, we extend intermediate solutions for child nodes to an intermediate solution for the parent node.
The existence of an intermediate solution for the root node, implies the existence of a solution for the original problem instance (Lemma~\ref{lem:plan-implies-yes-instance}).
Then, in Lemma~\ref{lem:complexity-btr-plan}, we establish the complexity of computing an intermediate solution, thereby completing the proof of Theorem~\ref{thm:treewidth}.
%
%
Roughly speaking, for every subset $T$ of the set $S$ of steps, each bag $\mathcal{B}$ in the tree decomposition of $D$, and each step $x$ in $T$, we keep track of which user in $\mathcal{B}$, if any, $x$ is to be assigned to, and otherwise what relation the user assigned to $x$ should have to the users in $\mathcal{B}$. Before proceeding further, we introduce some definitions and notation.

Let us say that $u > v$ if $v < u$, and $u \sim v$ if neither $u < v$
nor $v < u$. Define the \emph{relation} of $v$ to $u$, a function
$\phi(v,u)$ from the set of all pairs of users to the set of three
symbols $[<], [>], [\sim]$, as follows:
\begin{equation*}
\phi(v,u) = \begin{cases} [<] \text{ if } v < u \\
[>] \text{ if } v > u\\
[\sim] \text{ if } v \sim u.\end{cases}
\end{equation*}

For each bag ${\cal B} = \{u_1, u_2, \dots, u_p\}$  in $\mathcal{X}$,
and each user $v \notin \mathcal{B}$, define the \emph{relation} of
$v$ to $\mathcal{B}$, $\mathcal{R}(v, \mathcal{B})$ to be the ordered
tuple $(\phi(v,u_1), \dots, \phi(v, u_p))$.

\begin{definition}
Given a workflow instance $\mathcal{I}$ with constraint graph $G=(S,E)$, a bag $\mathcal{B}$ in the nice tree decomposition of $(U,<)$, a set of steps $T$ and a function $R : T \rightarrow {\cal B} \cup \{[<],[>],[\sim]\}^{|\mathcal{B}|}$, we say $\pi: T \rightarrow U$ is a \emph{$(\mathcal{B},T,R)$-plan} if the following conditions are satisfied:
\begin{enumerate}
 \item $\pi(x)\in L(x)$ for each $x \in T$;
 \item if there is an edge between $x$ and $y$ in $G[T]$, then $\pi(x) \neq \pi(y)$;
 \item if there is an arc from $x$ to $y$ in $G[T]$, then $\pi(x) < \pi(y)$;
 \item for each step $x$, $\pi(x)$ is either a user in ${\cal B}$ or a user in a descendant of ${\cal B}$;
 \item for any $x \in T$, $u \in {\cal B}$, $\pi(x)=u$ if and only if $R(x)=u$;
 \item if $R(x) \notin {\cal B}$, then $\mathcal{R}(\pi(x), \mathcal{B}) = R(x)$.
\end{enumerate}
\end{definition}

$R$ provides a partial allocation of users in $\mathcal{B}$ to steps in $T$; where no user is allocated, $R$ identifies the relationships that must hold between the user that is subsequently allocated to the task and those users in $\mathcal{B}$.
The existence of a $(\mathcal{B},T,R)$-plan means that we can extend $R$ to a full plan $\pi$ by traversing the nice tree decomposition.

We may now define the function that is central to our dynamic programming approach.
For every bag ${\cal B}$ in the tree decomposition of $D$, every subset $T $ of $S$, and every possible function $R: T \rightarrow {\cal B} \cup \{[<],[>],[\sim]\}^{|\mathcal{B}|}$, define $F({\cal B}, T, R)  = \textsc{True}$ if there exists a $(\mathcal{B},T,R)$-plan and \textsc{False} otherwise.



\begin{lemma}\label{lem:plan-implies-yes-instance}
Let ${\cal B}_0$ be the root node in the nice tree decomposition of $D$. Then
$\mathcal{I}$ is a {\sc Yes}-instance if and only if there exists a function
$R: S \rightarrow {\cal B}_0 \cup \{[<],[>],[\sim]\}^{|\mathcal{B}_0|}$
such that $F({\cal B}_0, S, R)  = \textsc{True}$.
\end{lemma}
\begin{proof}
By the first three conditions on  $F({\cal B}_0, S, R)$ being {\sc True} and the definition of the constraint graph $G$, it is clear that if $F({\cal B}_0, S, R)  = \textsc{True}$ for some $R$ then we have a {\sc Yes}-instance.
So now suppose $\mathcal{I}$ is a {\sc Yes}-instance, and let $\pi: S \rightarrow U$ be a valid plan. Then for each $x \in S$, let $R(x)=\pi(x)$ if $\pi(x)\in {\cal B}_0$, and otherwise, let $R(x) = \mathcal{R}(\pi(x),\mathcal{B})$.
Then observe that all the conditions on $F({\cal B}_0, S, R)$ being {\sc True} are satisfied and therefore $F({\cal B}_0, V, R) = \textsc{True}$.
\end{proof}


\begin{lemma}\label{lem:complexity-btr-plan}
We can compute $F({\cal B}, T, R)$ for every bag ${\cal B}$ in ${\cal X}$, every $T \subseteq S$, and every  $R: T \rightarrow {\cal B} \cup \{[<],[>],[\sim]\}^{|{\cal B}|}$ in time $O(nk4^k(r+2+3^{r+1})^k)$.
\end{lemma}

\begin{proof}
We will start by constructing, in advance, a matrix ${\cal L}=[L_{s,u}]_{s\in S, u\in U}$ such that $L_{s,u}=1$ if $u\in L(s)$ and $L_{s,u}=0$, otherwise. This will take time $O(kn)$.
Let $\mathcal{B}$ be in ${\cal X}$, $T$ a subset of $S$, and $R$ a function from $T$ to ${\cal B} \cup \{[<],[>],[\sim]\}^{|{\cal B}|}.$
Recall that every bag $\mathcal{B}$ is either a leaf node, a join node, a forget node or an introduce node. We will consider the four possibilities separately.


\paragraph{\bf $\mathcal{B}$ is a leaf node.}

Since ${\cal B}$ has no descendants, $F({\cal B}, T, R) = \textsc{False}$ unless $R(x) \in {\cal B}$ for every $x \in T$.
So now we may assume $R(x) \in {\cal B}$ for all $x$. But then the only possibility for a $({\cal B}, T, R) $-plan is one in which $\pi(x)=R(x)$ for all $x$. Therefore we may check, in  time $O(k^2)$,
whether this plan satisfies the (first three) conditions on $F({\cal B}, T, R)$ being {\sc True}. (Use matrix $\cal L$ to check that $\pi(x) \in L(x)$ for all $x\in T$.)
If it does,  $F({\cal B}, T, R) = \textsc{True}$, otherwise  $F({\cal B}, T, R) = \textsc{False}$.


\medskip

For the remaining cases, we may assume that $F({\cal B}', T, R)$ has been calculated for every child of ${\cal B}'$ of ${\cal B}$ and every possible $T,R$.


\paragraph{\bf ${\cal B}$ is a forget node.}

Let ${\cal B}' = \{u_1, u_2, \dots, u_p\}$ be the child node of ${\cal B}$ and assume without loss of generality that
${\cal B} = \{u_1, u_2, \dots, u_{p-1}\}$.
For $i \in [p-1]$, let $X_i$ be the set of steps in $T$ with $R(x)=u_i$.

Suppose that $\pi$ is a $({\cal B}, T, R)$-plan.
Then let $R': T \rightarrow {\cal B}' \cup \{[<],[>],[\sim]\}^{|{\cal B}'|}$ be the function such that $R'(x)=\pi(x)$ if $\pi(x) \in {\cal B}'$, and $R'(x)={\cal R}(\pi(x),{\cal B}')$ if $\pi(x) \notin {\cal B}'$. It is clear that $F({\cal B}', T, R') = \textsc{True}$. Now we show some properties of $R$.

Firstly, since $\pi$ is a $({\cal B}, T, R)$-plan, it must be the case that $\pi(x)=R(x)$ if $R(x) \in {\cal B}$ and therefore $R'(x)=R(x)$ if $R(x) \in {\cal B}$.
Secondly, since $\pi$ is a $({\cal B}, T, R)$-plan and $u_p \notin \mathcal{B}$, it must be the case that $\pi(x)=u_p$ only if $R(x) = {\cal R}(u_p, {\cal B})$. Therefore $R'(x)=u_p$ only if $R(x)= {\cal R}(u_p, {\cal B})$.
Finally, for $x \in T$ with $R'(x) \notin {\cal B'}$, let $R(x)=(x_{u_1}, x_{u_2}, \dots, x_{u_{p-1}})$ and let $R'(x) = (x'_{u_1}, x'_{u_2}, \dots, x'_{u_p})$. Since $\pi$ is a $({\cal B}, T, R)$-plan and a $({\cal B}', T, R')$-plan, we must have that $x_{u_i} = \phi(\pi(x),u_i) = x'_{u_i}$ for all $i \in [p-1]$. That is $R(x)$ and $R'(x)$ are the same except that $R'(x)$ has the extra co-ordinate $x'_{u_p}$. It follows that to obtain $R'$ from $R$, we merely need to guess which
$x$ with $R(x) = \mathcal{R}(u_p, \mathcal{B})$ are assigned to $u_p$ by
$R'$, and for all other $x$, what the value of $x_{u_p}$ should be.

Therefore, in order to calculate $F({\cal B}, T, R)$, we may do the following: Try every possible way of partitioning $T \setminus (X_1 \cup X_2 \cup \dots \cup X_{p-1})$ into four sets $X_p, X_{<}, X_{>}, X_{\sim}$, subject to the constraint that $x \in X_p$ only if $R(x)={\cal R}(u_p, {\cal B})$.
For each such partition, construct a function $R': T \rightarrow {\cal B}' \cup \{[<],[>],[\sim]\}^{|{\cal B}'|}$ such that
\begin{enumerate}
 \item $R'(x)=R(x)$ if $R(x)\in{\cal B}$.
 \item $R'(x)=u_p$ if $x \in X_p$.
 \item For all other $x$, let $R(x) = (x_{u_1}, x_{u_2}, \dots, x_{u_{p-1}})$. Then $R'(x) = (x'_{u_1}, x'_{u_2}, \dots, x'_{u_p})$, where $x'_{u_i} = x_{u_i}$ for all $i \in [p-1]$, and $x'_{u_p} = [<]$ if $x \in X_{<}$, $x'_{u_p} = [>]$ if $x \in X_{>}$, and $x'_{u_p} = [\sim]$ if $x \in X_{\sim}$.
\end{enumerate}
and check the value of $F({\cal B}', T, R')$.

By the above argument, if $F({\cal B}, T, R) = \textsc{True}$ then it must be the case that $F({\cal B}', T, R') = \textsc{True}$ for one of the $R'$ constructed in this way. Therefore if $F({\cal B}', T, R') = \textsc{False}$ for all such $R'$, we know that $F({\cal B}, T, R) = \textsc{False}$. Otherwise, if $F({\cal B}', T, R') = \textsc{True}$ for some $R'$, let $\pi$ be a $({\cal B}', T, R')$-plan, and observe that by construction of $R'$ and $X_p$, $\pi$ is a $({\cal B}, T, R)$-plan as well. Therefore $F({\cal B}, T, R) = \textsc{True}$.

Finally, observe that there are at most $4^k$ possible values of $R'$ to check and each $R'$ can be constructed in time $O(k)$, and therefore we can calculate $F({\cal B}, T, R)$ in time $O(k4^k)$.


\paragraph{\bf ${\cal B}$ is an introduce node.}

Let ${\cal B} = \{u_1, u_2, \dots, u_p\}$, let ${\cal B}'$ be the child node of ${\cal B}$ and assume without loss of generality that
${\cal B}' = \{u_1, u_2, \dots, u_{p-1}\}$.
Let $X_p \subseteq T$ be the set of all $x \in T$ with $R(x) = u_p$, and let $T' = T\setminus X_p$.
Define a function $R': T' \rightarrow {\cal B}' \cup \{[<],[>],[\sim]\}^{|{\cal B}'|}$ as follows:

\begin{enumerate}
 \item $R'(x) = R(x)$ if $R(x) \in {\cal B}'$.
 \item For all other $x$, let $R(x) = (x_{u_1}, x_{u_2}, \dots, x_{u_p})$. Then set $R'(x) = (x'_{u_1}, x'_{u_2}, \dots, x'_{u_{p-1}})$, where $x'_{u_i} = x_{u_i}$ for all $i \in [p-1]$.
\end{enumerate}

We will now give eight conditions which are necessary for $F({\cal B}, T, R) = \textsc{True}$. We will then show that these conditions collectively are sufficient for $F({\cal B}, T, R) = \textsc{True}$. Since each of these conditions can be checked in time $O(k^2)$, we will have that $F({\cal B}, T, R)$ can be calculated in time $O(k^2)$.

\emph{Condition $1$: $L_{x,u_p}=1$ for each $x \in X_p$}. This condition is clearly necessary, as for every $({\cal B}, T, R)$-plan $\pi$ we have $\pi(x)=u_p$.

\emph{Condition $2$: $X_p$ is an independent set in $G$}. Since in any $({\cal B}, T, R)$-plan, all steps in $X_p$ must be assigned the same user, any arc  or edge between steps in $X_p$ will not be satisfied.

\emph{Condition $3$: If there exists $x \in X_p$, $y \notin X_p$ with an arc from $y$ to $x$ in $G$, then either $R(y) = u_i$ for some $u_i \in {\cal B}'$ with $u_i < u_p$, or $R(y) = (y_{u_1}, \dots , y_{u_p})$ with $y_{u_p} = [<]$.}
For if not, then any $({\cal B}, T, R)$-plan will assign $y$ to a user $v$ such that $v > u_p$ or $v \sim u_p$, and the arc $yx$ will not be satisfied.

\emph{Condition $4$: If there exists $x \in X_p$, $y \notin X_p$ with an arc from $x$ to $y$ in $G$, then either $R(y) = u_i$ for some $u_i \in {\cal B}'$ with $u_i > u_p$, or $R(y) = (y_{u_1}, \dots , y_{u_p})$ with $u_p = [>]$.}
The proof is similar to the proof of Condition $3$.

\emph{Condition $5$: If there exists $y \notin X_p$ such that $R(y) = (y_{u_1}, \dots , y_{u_p})$ with $y_{u_p} = [<]$, then there must exist $u_i \in {\cal B}'$ with $y_{u_i} = [<]$ and $u_i < u_p$.}
For suppose there is a $({\cal B}, T, R)$-plan $\pi$, and let $v=\pi(y)$. Note that $v$ must be in a descendant of ${\cal B}$ but not in ${\cal B}'$. Therefore  ${\cal B}'$ separates $v$ from $u_p$ in $D$, for any $v$ in a descendant of ${\cal B}$. (This follows from Lemma \ref{lem:connected} where $Y$ is the vertices of a path between $v$ and $u_p$). Then by Lemma \ref{lem:separator}, as $v < u_p$ there exists $u_i \in {\cal B}'$ with $v < u_i < u_p$.
Therefore $y_{v_i} = [<]$.

\emph{Condition $6$: If there exists $y \notin X_p$ such that $R(y) = (y_{u_1}, \dots , y_{u_p})$ with $y_{u_p} = [>]$, then there must exist $u_i \in {\cal B}'$ with $y_{u_i} = [>]$ and $u_i > u_p$.}
The proof is similar to the proof of Condition $5$.

\emph{Condition $7$: If there exists $y \notin X_p$ such that $R(y) = (y_{u_1}, \dots , y_{u_p})$ with $y_{u_p} = [\sim]$, then there is no $u_i \in {\cal B}$ such that $y_{u_i} = [<]$ and $u_i < u_p$, or  $y_{u_i} = [>]$ and $u_i > u_p$.}
For suppose there is a $({\cal B}, T, R)$-plan $\pi$, and let $v=\pi(y)$. Suppose for a contradiction that there exists  $u_i \in {\cal B}$ such that $y_{u_i} = [>]$ and $u_i >u_p$. (The case $y_{u_i} = [<]$ and $u_i < u_p$ is handled similarly). Then $v > u_i$ and so by transitivity, $v > u_p$. But this is a contradiction as $y_{u_p} = [\sim]$.

\emph{Condition $8$: $F({\cal B}', T', R') = \textsc{True}$.}
For suppose $\pi$ is a $({\cal B}, T, R)$-plan. Then observe that by construction of $R'$, $\pi$ restricted to $T'$ is a $({\cal B}', T', R')$-plan.

It now remains to show that if Conditions $1$-$8$ hold then $({\cal B}, T, R) = \textsc{True}$. Let $\pi'$ be a $({\cal B}', T', R')$-plan whose existence is guaranteed by Condition $8$, and let $\pi$ be the extension of $\pi'$ to $T$ in which $\pi(x)=u_p$ for all $x \in X_p = T \setminus T'$.
Since $\pi'$ is a $({\cal B}', T', R')$-plan, $\pi(x)\in L(x)$ for all $x \in T'$, and by Condition $1$, $\pi(x) \in L(x)$ for all $x \in X_p$.
For every $x$ with $R(x) \in {\cal B}$, we have that $\pi(x)=R(x)$ by the fact that $\pi'$ is a $({\cal B}', T', R')$-plan and $R(x) = u_p$ for all $x \in X_p$.

Now consider $x$ with $R(x) \notin {\cal B}$. Then let $R(x) = (x_{u_1}, x_{u_2}, \dots, x_{u_p})$. By construction of $R'$ and the fact that $\pi'$ is a $({\cal B}', T', R')$-plan, $\phi(\pi(x), u_i)=x_{u_i}$ for $i \in [p-1]$. Suppose $x_{u_p} = [<]$. Then by Condition $5$, there exists $u_i \in {\cal B}'$ with $x_{u_i} = [<]$ and $u_i < u_p$. Therefore $\pi(x) < u_i$ and so $\pi(x) < u_p$. Therefore $\phi(\pi(x), u_i)=[<]$. Similarly, using Condition $6$, if $x_{u_p}= [>]$ then $\phi(\pi(x), {\cal B}) = [>]$. If $\phi(\pi(x), u_i)=[\sim]$ then by Condition $7$ there is no $u_i \in {\cal B}'$ with $\pi(x) > u_i > u_p$ or $\pi(x) < u_i < u_p$. Then by Lemma \ref{lem:separator}, $\pi(x)\sim u_p$ and so $\phi(\pi(x), u_p) = [\sim]$.
In each case we have that $\phi(\pi(x), u_p)=x_{u_p}$ and so ${\cal R}(\pi(x), {\cal B}) = R(x)$.

It is clear that for each step $x$, $\pi(x)$ is either in $\mathcal{B}$ or in a descendant of $\mathcal{B}$.
It remains to show that the arcs and edges in $G[T]$ are satisfied by $\pi$.

 As $\pi'$ is a $({\cal B}', T', R')$-plan, every arc and edge in $G[T']$ is satisfied by $\pi$. By Condition $2$ there are no edges and arcs within $G[X_p]$. It remains to show that the arcs and edges between $X_p$ and $T'$ are satisfied by $\pi$. Consider an edge  between $x \in X_p$ and $y \in T'$. Since $\pi(x)=u_p$, and $\pi(y)\neq u_p$ (since $u_p$ does not appear in ${\cal B}'$ or any descendant of ${\cal B}'$ by definition of a tree decomposition), this edge is satisfied. Now suppose there is an arc from $y \in T'$ to $x \in X_p$ . By Condition $3$, either $\pi(y)=R(y)=u_i$ with $u_i < u_p$, or $y_{u_p}=[<]$, in which case $\pi(y) < u_p$ (as we have shown $\phi(\pi(y, u_p)=y_{u_p}$). In either case $\pi(y) < \pi(x)$ and so the arc is satisfied. Similarly, if there is an arc from $x \in X_p$ to $y \in S'$, then by Condition $4$ $\pi(y) > \pi(x)$ and the arc is satisfied.

Thus $\pi$ satisfies all the conditions of a $(\mathcal{B},T,R)$-plan and so $F({\cal B}, T, R)=\textsc{True}$.


\paragraph{\bf ${\cal B}$ is a join node.}
Let ${\cal B}', {\cal B}''$ be the two child nodes of ${\cal B}$, and recall that ${\cal B}'$ and ${\cal B}''$ contain the same users as ${\cal B}$.
Let $X$ be the set of all $x \in T$ with $R(x) \in {\cal B}$.

Let $\pi$ be a $({\cal B}, T, R)$-plan. Then let $X'$ be the set of all $x \in T \setminus X$ such that $\pi(x)=v$ for some $v$ in a descendant of ${\cal B}'$, and let $X''$ be the set of all $x \in T \setminus X$ such that $\pi(x)=v$ for some $v$ in a descendant of ${\cal B}''$. (Observe that $X,X',X''$ is a partition of $T$.) Let $T'=X \cup X'$ and let $R'$ be the function $R$ restricted to $T'$.
Similarly let $T''=X \cup X''$ and let $R''$ be the function $R$ restricted to $T''$.
Then observe that $F({\cal B}', T', R')=\textsc{True}$ and $F({\cal B}'', T'', R'')=\textsc{True}$.

Now consider an arc from $x \in X'$ to $y \in X''$. Then $\pi(x) < \pi(y)$. Since ${\cal B}$ separates $\pi(x)$ from $\pi(y)$ (by Lemma \ref{lem:connected} with $Y$ the set of vertices on a path between $\pi(x)$ and $\pi(y)$), there must exist $u_i \in {\cal B}$ such that $\pi(x) < u_i < \pi(y)$. Therefore $x_{u_i}=[<]$ and $y_{u_i}=[>]$.
Similarly, if there is an arc from $y \in X''$ to $x \in X'$ then there exists $u_i \in {\cal B}$ with $x_{u_i}=[<]$ and $y_{u_i}=[>]$.

We therefore have that if $F({\cal B}, T, R)=\textsc{True}$, then there exists a partition $X',X''$ of $T \setminus X$ such that $F({\cal B}', T', R')=\textsc{True}$ and $F({\cal B}'', T'', R'')=\textsc{True}$ (where $T',T'',R',R''$ are as previously defined) and for any arc from $x \in X'$ to $y \in X''$, there exists $u_i \in {\cal B}$ with $x_{u_i}=[<]$ and $y_{u_i}=[>]$ (and similarly for arcs from $y \in X''$ to $x \in X'$). We now show that the converse is true.

Suppose these conditions hold, and let $\pi'$ be a $({\cal B}', T', R')$-plan and $\pi''$ a $({\cal B}'', T'', R'')$-plan.
Note that for all $x \in X$, $\pi'(x)=R(x)=\pi''(x)$.
Let $\pi$ be the assignment on $S$ made by combining $\pi'$ and $\pi''$, i.e. $\pi(x) = \pi'(x)=\pi''(x)$ for $x \in X$, $\pi(x)=\pi'(x)$ for $x \in X'$, and $\pi(x)=\pi''(x)$ for $x \in X''$.

Observe that by definition of $\pi'$ and $\pi''$, $L_{x, \pi(x)} = 1$ for all $x \in T$, $\pi(x)=R(x)$ if $R(x) \in {\cal B}$, and otherwise ${\cal R}(\pi(x),{\cal B})=R(x)$. Any edges and arcs in $G[X \cup X']$ are satisfied by $\pi$, by definition of $\pi'$, and any edges and arcs in $G[X \cup X'']$ are satisfied by $\pi$, by definition of $\pi''$. It remains to consider the edges and arcs between $X'$ and $X''$. Since the tasks in $X'$ and $X''$ are assigned to disjoint sets of users (by Lemma \ref{lem:connected}), any edge  between and $X'$ and $X''$ is satisfied. If there is an arc from $x \in X'$ to $y \in X''$, then by our assumption there exists $u_i \in {\cal B}$ with $x_{u_i}=[<]$ and $y_{u_i}=[>]$. Therefore $\pi(x)< u_i < \pi(y)$, and therefore $\pi(x)< \pi(y)$, and so the arc is satisfied. A similar argument applies when there is an arc from $y \in X''$ to $x \in X'$.

Since there are at most $2^{|T|}$ possible ways to partition $T \setminus X$ into $X'$ and $X''$, we can calculate $F({\cal B}, T, R)$ in $O(2^k)$ time.

\medskip

The above bounds show that, provided all the values for descendants of ${\cal B}$ have been computed, $F({\cal B}, T, R)$ can be calculated in time $O(k4^k)$, for each possible ${\cal B}, T$ and $R$. It remains to count the number of possible values of ${\cal B}, T$ and $R$.
There are at most $4n$ values of ${\cal B}$. Calculating $F({\cal B},T,R)$ for every $T$ and $R$ can be viewed as calculating $F$ for
every function $R^*: S \rightarrow {\cal B} \cup \{[<],[>],[\sim]\}^{|{\cal B}|}\cup \{0\}$, $T$ being defined as the set of steps not mapped to $0$. Finally, for each step $x$ in $S$ there are $r+2+3^{r+1}$ possible values for $R^*(x)$ and therefore $(r+2+3^{r+1})^k$ possible values for $R^*$. Therefore the total number of possible values of $F({\cal B}, T, R)$ is $O(n(r+2+3^{r+1})^k)$, and so every value $F({\cal B}, T, R)$ can be calculated in time $O(nk4^k(r+2+3^{r+1})^k)$.\qed

\end{proof}

\section{Hardness}\label{sec:hardness}

The main theorem of this section establishes a lower bound for the complexity of the workflow satisfiability problem. In fact, we show that in general, the trivial $O(n^k)$ algorithm is nearly optimal. Our result assumes the Exponential Time Hypothesis ({\ETH}) of Impagliazzo, Paturi, and Zane~\cite{ImPaZa01}: that is, we assume that there is no $2^{o(n)}$-time algorithm for $n$-variable 3-SAT.

\begin{theorem}\label{thm:wsp hardness}
{\wsp} cannot be solved in time $f(k)n^{o({\frac {k} {\log k}})}$ unless {\ETH} fails, where $f$ is an arbitrary function, $k$ is the number of steps and $n$ is the number of users. This results holds even if the full graph of $(U,<)$ is 2-degenerate.
\end{theorem}

The proof of Theorem~\ref{thm:wsp hardness} can be found in the appendix.
It is well-known (see, e.g., \cite{FluGro06}) that {\ETH} is stronger than the widely believed complexity hypothesis $\text{W[1]} \neq \text{FPT}$. Thus, we have the following:

\begin{corollary}
{\wsp} is not FPT unless $\text{W[1]} = \text{FPT}$. This results holds even if the full graph of $(U,<)$ is 2-degenerate.
\end{corollary}

This corollary proves that while the class of treewidth bounded graphs is sufficiently special to imply an FPT algorithm, considering the more general class of graphs of bounded degeneracy does not make the problem any easier.

\section{Concluding Remarks}\label{sec:conclusion}

The main contribution of this paper is the development of the first FPT algorithm for \wsp{}, where $<$ is a (transitive) relation on the set of users. Unlike  {\sc WSP}($=,\neq$) which is FPT in the general case,  {\sc WSP}($=,\neq,<$) is not FPT unless W[1]=FPT, which is highly unlikely. In fact, under a stronger hypothesis ({\ETH}) we have shown that we even cannot have an algorithm significantly faster than the trivial brute-force algorithm. Thus, it is natural to identify special cases of {\sc WSP}($=,\neq,<$) that are in FPT and of practical relevance. We have done this by  restricting the reduced graph $D$ of $(U,<)$ to lie in the class of graphs of bounded treewidth. We believe that this restriction on treewidth holds for many user hierarchies that arise in practice. On the other hand, we have also shown that the restriction of the reduced (or even full) graph to the class of 2-degenerate graphs does not reduce the complexity of the problem.

Our FPT algorithm is efficient for small values of the number of steps $k$ and the treewidth $r$ of $D$ (we may view $k+r$ as a combined parameter). However, it is quite often the case that the first FPT algorithm for a parameterized problem is not efficient except for rather small values of the parameter, but subsequent improvements bring about an  FPT algorithm efficient for quite large values of the parameter \cite{FluGro06,Nie06}. We believe that a more efficient FPT algorithm for \wsp\ may be possible and we hope to be able to report progress in this area.

One natural extension of this work is to consider the preorder generated from an authorization policy, where $u \sqsubseteq u'$ iff the set of steps for which $u$ is authorized is a subset of the set of steps for which $u'$ is authorized.
This ordering is weaker than that defined in Sec.~\ref{sec:satisfiability} and used throughout the rest of the paper, which required that the set of steps for which $u$ is authorized to be a strict subset of those for which $u'$ is authorized.
Hence, we may have $u \sqsubseteq u'$ and $u' \sqsubseteq u$ but $u \ne u'$.
In fact, such an ordering defines sets of users that are \emph{indistinguishable}, in the sense that they are authorized for the same set of steps.
Hence, we might reasonably consider WSP$(=,\ne,\sqsubset,\sim,\nsim)$, where $u \sim u'$ if $u$ and $u'$ are indistinguishable.
Of course, the graph of $\sqsubseteq$ is not acyclic, as cycles of length two will exist between any pair of indistinguishable users, so new techniques may be required to determine whether this problem is FPT or not.


\medskip

\paragraph{Acknowledgment}
This research was partially supported by an International Joint grant of the Royal Society.

\bibliography{refs}
\bibliographystyle{splncs03}

\clearpage

\appendix\section{Proof of Theorem~\ref{thm:wsp hardness}}

In order to prove Theorem~\ref{thm:wsp hardness}, we first consider the following problem and prove the following lemma.

\begin{center}
\begin{boxedminipage}{.8\textwidth}
\parnamedefn{{\sc SubTDAG isomorphism}}{Transitive acyclic digraphs $D=(V_D,A_D)$
and $R=(V_R,A_R)$,
a subset $W_R=\{w_1,\ldots, w_{|W_R|}\}$ of $V_R$, and disjoint subsets $W_{D,1},\ldots ,W_{D,|W_R|}$ of  $V_D$.
}{$\vert V_R\vert$}{Is there an injection $\gamma:V_R\rightarrow V_D$ such that $\gamma(w_i)\in W_{D,i}$ for each $i\in [|W_R|]$, and
for every $(u,v)\in A_R$, $(\gamma(u),\gamma(v))\in A_D$?}
\end{boxedminipage}
\end{center}

\begin{lemma}\label{lem:dag iso hardness}
{\subdagiso} cannot be solved in time $f(k)n^{o({\frac {k} {\log k}})}$ where $f$ is an arbitrary function, $n$ is the number of vertices in $D$ and $k$ is the number of vertices in $R$, unless {\ETH} fails. This result holds even if $D$ and $R$ are 2-degenerate.
\end{lemma}

\noindent
To prove Lemma~\ref{lem:dag iso hardness}, we start by considering the following problem and a lemma by Marx~\cite{marx-toc-treewidth}.

\

\begin{center}
\begin{boxedminipage}{.8\textwidth}
\decnamedefn{{\sc Partitioned Subgraph Isomorphism (PSI)}}{Undirected graphs $H=(V_H,E_H)$ and $G=(V_G=\{g_1,\dots,g_l\},E_G)$, and a partition of $V_H$ into (disjoint) subsets $W_{H,1},\ldots ,W_{H,l}$.}
{ Is there an injection $\phi:V_G\rightarrow V_H$ such that for every $i\in [l]$, $\phi(g_i)\in W_{H,i}$ and for every $(g_i,g_j)\in E_G$, $(\phi(g_i),\phi(g_j))\in E_H$?
}
\end{boxedminipage}
\end{center}

\begin{lemma}{\sc (}Corollary 6.3, {\sc \cite{marx-toc-treewidth})}\label{lem:cis hardness}
{\sc Partitioned Subgraph Isomorphism} cannot be solved in time $f(k)n^{o(\frac{k} {\log k})}$ where $f$ is an arbitrary function, $k$ is the number of edges in $G$ and $n$ is the number of vertices in $H$, unless {\ETH} fails.
\end{lemma}

\noindent
{\bf Proof of Lemma}~\ref{lem:dag iso hardness}. The proof is by a reduction from the {\sc Partitioned Subgraph Isomorphism} problem. We assume that we have an instance of {\sc PSI} as described in the formulation of the problem above.
We assume, without loss of generality, that there are no isolated vertices in $G$. 
Recall that the vertices of $G$ are $g_1,\dots, g_l$ and let $W_{H,1}=\{x(11),\ldots ,x(1r_1)\}, \ldots ,W_{H,l}=\{x(l1),\dots ,x(lr_l)\}.$
We now construct an instance of {\subdagiso}. The digraph $R$ is obtained from $G$ by subdividing every edge and orienting all edges towards the new vertices.  The vertex subdividing an edge $g_ig_j$ will be denoted by $g_{ij}$
and so $R$ will have arcs $(g_i,g_{ij})$ and $(g_j,g_{ij})$. Similarly, $D$ is obtained from $H$ by subdividing every edge and orienting all edges towards the new vertices. The vertex subdividing an edge
$x(i\tau_i)x(j\tau_j)$ will be denoted by $x(i\tau_i,j\tau_j)$. It is easy to verify that $D$ and $R$ are both 2-degenerate acyclic digraphs and both are transitive because they do not have directed paths of length 2. Let $W_R=V_G$ and $W_{D,i}=W_{H,i}$ for each $i\in [l].$ We claim that $(G,H,W_{H,1},\ldots ,W_{H,l})$ is a {\Yes}-instance of {\sc PSI} if and only if $(D,R,W_R,W_{D,1},\ldots ,W_{D,l})$ is a {\Yes}-instance of {\subdagiso}.\\

Suppose that our instance of {\sc PSI} is a {\Yes}-instance and let $\phi$ be the required injection.
By definition, $\phi(g_i)=x(i\tau_i)$, where $\tau_i\in [r_i]$, for each $i\in [l]$. Let $\gamma: V_R\rightarrow V_D$ be defined as follows: $\gamma(g_i)=x(i\tau_i)$ for each $i\in [l]$ and $\gamma(g_{ij})=x(i\tau_i,j\tau_j)$. Since $g_ig_j\in E_G$ implies $\phi(g_i)\phi(g_j)\in E_H$ and by the definition of $\gamma$, if $(g_i,g_{ij})\in A_R$ then $(\gamma(g_i),\gamma(g_{ij}))\in A_D.$  Thus, our instance of {\subdagiso} is a {\Yes}-instance, too.

Now suppose that the instance of  {\subdagiso} is a {\Yes}-instance and $\gamma: V_R\rightarrow V_D$ is the corresponding injection such that $\gamma(g_i)=x(i\tau_i)$, where $\tau_i\in [r_i]$, for each $i\in [l]$.
By definition of $\gamma$, $(g_i,g_{ij})\in A_R$ implies $(x(i\tau_i)\gamma(g_{ij}))\in A_D$ and $(g_j,g_{ij})\in A_R$ implies $(x(j\tau_j)\gamma(g_{ij}))\in A_D$. By the construction of $D$, the above implies that $\gamma(g_{ij})=x(i\tau_i,j\tau_j)$. Now define an injection $\phi: G\rightarrow H$ as follows: $\phi(g_i)=\gamma(g_i)=x(i\tau_i)$ for each $i\in [l]$. The requirement that $g_ig_j\in E_G$ implies $\phi(g_i)\phi(g_j)\in E_H$
follows  from the fact that $\gamma(g_{ij})=x(i\tau_i,j\tau_j)$. Thus, the instance of {\sc PSI} is a {\Yes}-instance, too.

Let $k_G$ be the number of edges in $G$ and $n_H$ the number of vertices in $H$. Recall that $k$ is the number of vertices in $R$ and $n$ is the number of vertices in $D$. By construction of $R$ and $D$ and the assumption that $G$ has no isolated vertices, $k=|E_G|+|V_G|=\Theta(k_G)$
and $n=n_H+|A_H|=O(n_H^2)$.


An algorithm for {\subdagiso} running in time $f(k)n^{o(\frac{k}{\log k})}$ implies an algorithm running in time $f(k_G)n_H^{o(\frac{k_G}{\log k_G})}$ for {\sc PSI}, which along with Lemma~\ref{lem:cis hardness} completes the proof of the lemma.\qed \vspace{5 pt}

\noindent
{\bf Proof of Theorem}~\ref{thm:wsp hardness}. The proof is by a reduction from the {\subdagiso} problem. Let $(D,R,W_R,W_{D,1},\ldots ,W_{D,|W_R|})$ be an instance of {\subdagiso}. We construct an instance of {\wsp} as follows. We define the set $U$ of users to be $V_D$ and the set $S$ of steps to be $V_R$. For every step $w_i\in W_R$, $L(w_i)=W_{D,i}$,
 and for every step $s\in S\setminus W_R$, $L(s)=U$.


We define the relation $<$ on $U$ as follows. For every $x,y\in U$, $x<y$ if and only if $x\neq y$ and there is a arc from $x$ to $y$ in $D$. For every arc $(u,v)\in A_R$, we add a constraint $(<,u,v)$ and for every pair $u,v$ of distinct non-adjacent vertices of $R$, we add a constraint $(\neq,u,v)$. Let the instance of {\wsp} thus constructed be $\cal I$. We claim that $(D,R,W_R,W_{D,1},\ldots ,W_{D,|W_R|})$is a {\Yes} instance of {\subdagiso} iff $\cal I$ is a {\Yes} instance of {\wsp}.

Suppose that $(D,R,W_R,W_{D,1},\ldots ,W_{D,|W_R|})$ is a {\Yes}-instance of {\subdagiso} and let $\gamma$ be a required injection for this instance. We define a plan $\pi$ as $\pi(v)=\gamma(v)$ for every $v\in S$. It is easy to see that $\pi$ is an valid plan for $\cal I$.

Conversely, suppose that $\cal I$ is a {\Yes}-instance of {\wsp} and let $\pi$ be a valid plan for this instance. We define a function $\gamma:V_R\rightarrow V_D$ as follows. For every $u\in V_R$, we set $\gamma(u)=\pi(u)$. It remains to verify that $\gamma$ is a required injection for the instance $(D,R,W_R,W_{D,1},\ldots ,W_{D,|W_R|})$. We first show that $\gamma$ is an injection. Suppose this were not the case and let $u$ and $v$ be two distinct vertices such that $\gamma(u)=\gamma(v)$. This implies that $\pi(u)=\pi(v)$. But then this assignment satisfies neither the constraint $(\neq,u,v)$ nor the constraint $(<,u,v)$, which is a contradiction. Hence, we conclude that $\gamma$ is indeed an injection. Now, consider an arc $(u,v)\in R$. Since $\pi$ is a valid plan, $\pi(u)<\pi(v)$, which implies that $\gamma(u)<\gamma(v)$, which by definition is possible only if $(\gamma(u),\gamma(v))\in A_D$. This completes the proof of correctness of the reduction.

It remains to apply Lemma~\ref{lem:dag iso hardness} to complete the proof of the theorem.
\qed

\end{document}